\documentclass{amsart}

\usepackage{amsmath}
\usepackage{amsthm}
\usepackage{amsfonts}
\usepackage{graphicx}
\usepackage{latexsym}
\usepackage{amssymb}
\usepackage{amscd}
\usepackage{booktabs}
\usepackage[small]{caption}
\usepackage[title]{appendix}

\newcommand{\field}[1]{\mathbb{#1}}
\newcommand{\C}{\field{C}}

\newcommand{\dS}{\field{S}}

\newcommand{\cA}{{\mathcal A}}
\newcommand{\cB}{{\mathcal B}}
\newcommand{\cC}{{\mathcal C}}

\newcommand{\cG}{{\mathcal G}}

\newcommand{\cP}{{\mathcal P}}

\newcommand{\cH}{{\mathcal H}}

\newcommand{\sbinom}[2]{\left[ \begin{array}{c} #1 \\ #2 \end{array} \right] }
\newcommand{\sbinomq}[2]{\sbinom{#1}{#2}_q }
\newcommand{\sbinomtwo}[2]{\sbinom{#1}{#2}_2 }

\title{Subspace Packings}

\author[T.\,Etzion]{Tuvi Etzion}
\address{
Tuvi Etzion\\
Technion, Haifa, 
Israel}
\email{etzion@cs.technion.ac.il}

\author[S.\,Kurz]{Sascha Kurz}
\address{
Sascha Kurz\\
University of Bayreuth\\ 
Bayreuth, Germany}
\email{sascha.kurz@uni-bayreuth.de}

\author[K.\,Otal]{Kamil Otal}
\address{
Kamil Otal\\
T\"{U}BITAK BILGEM UEKAE\\
Gebze Turkey} 
\email{kamil.otal@gmail.com}

\author[F.\,\"Ozbudak]{Ferruh \"Ozbudak\vspace*{7ex}}
\address{
Ferruh \"Ozbudak\\
Middle East Technical University\\
Ankara, Turkey}
\email{ozbudak@metu.edu.tr}

\newtheorem{theorem}{Theorem}
\newtheorem{lemma}[theorem]{Lemma}

\newtheorem{remark}[theorem]{Remark}
\newtheorem{construction}[theorem]{Construction}
\newtheorem{claim}[theorem]{Claim}
\newtheorem{corollary}[theorem]{Corollary}
\newtheorem{proposition}[theorem]{Proposition}
\newtheorem{definition}[theorem]{Definition}

\renewcommand{\leq}{\leqslant}
\renewcommand{\geq}{\geqslant}

\newcommand{\F}{\mathbb{F}}

\makeatletter
\DeclareRobustCommand{\sbinom}{\genfrac[]\z@{}}
\makeatother

\begin{document}

\begin{abstract}
The {Grassmannian} $\cG_q(n,k)$
is the set of all $k$-dimensional subspaces of the vector space~$\F_q^n$.
It is well known that codes in the Grassmannian
space can be used for error-correction in random network coding.
On the other hand, these codes are $q$-analogs of codes in the Johnson scheme, i.e.
constant dimension codes. These codes of the Grassmannian $\cG_q(n,k)$
also form a family of $q$-analogs of block designs and they are
called \emph{subspace designs}. The application of subspace codes has motivated
extensive work on the $q$-analogs of block designs.

In this paper, we examine one of the last families of $q$-analogs of block designs
which was not considered before. This family called \emph{subspace packings}
is the $q$-analog of packings. This family of designs was considered recently
for network coding solution for a family of multicast networks called the generalized
combination networks. A \emph{subspace packing} $t$-$(n,k,\lambda)^m_q$ is a set $\dS$
of $k$-subspaces from $\cG_q(n,k)$ such that each $t$-subspace of $\cG_q(n,t)$ is contained
in at most $\lambda$ elements of $\dS$. The goal of this work is to consider
the largest size of such subspace packings.
\end{abstract}

\maketitle

\thispagestyle{empty}



\section{Introduction}
\label{sec:introduction}

A \emph{subspace packing} $t-(n,k,\lambda)_q^m$ is a set $\dS$ of $k$-subspaces (called \emph{blocks}) of $\F_q^n$ such that
each $t$-subspace of $\F_q^n$ is contained in at most $\lambda$ blocks. Throughout this paper $m$ in this notation refers to {\it multiplicity}, which complies with the notation already used in \cite{EtZh18}. 
The definition of a subspace packing is a straightforward
definition for $q$-analog of packing for set. Moreover, subspace packings have found recently another
nice application in network coding. It was proved in~\cite{EtZh18} that the code formed from the dual
subspaces (of dimension $n-k$) of a subspace packing is exactly what is required for a scalar
solution for a family of networks called the \emph{generalized combination networks}. This family of networks
was used in~\cite{EtWa18} to show that vector network coding outperforms scalar linear network coding on multicast networks.
The interested reader is invited to look in these paper for the required definition and the proof of the mentioned results.
For the network coding solution of the generalized combination networks repeated codes are allowed.
But, throughout our exposition we will assume
that there are no repeated blocks in the packing. This is the usual convention in block design and coding theory.

Let $\cA_q(n,k,t;\lambda)$ be the maximum number of $k$-subspaces in a $t-(n,k,\lambda)_q^m$ subspace design.
Although there are some upper bounds on $\cA_q(n,k,t;\lambda)$ and analysis of subspace designs in~\cite{EtZh18} the topic was hardly considered.
In~\cite{EtZh18} the authors mainly considered the related network coding problems and
a general analysis of the quantity $\cA_q(n,k,t;\lambda)$. The dual subspaces and the related codes
were also considered in~\cite{EtZh18}. For lack of space we will quote results in~\cite{EtZh18},
but not write them explicitly. Subspace packings are $q$-analog of packing designs which were
extensively studied, see~\cite{MiMu92,SWY06} and references therein.
The goal of the current work  is to present
a comprehensive study of subspace packings and to learn their upper bounds and constructions.
For lack of space, we will present only a few interesting bounds which are not straight forward
generalizations. The other will be presented in the full extended version of this paper.

The rest of this paper is organized as follows.
In Section~\ref{sec:upper} upper bounds are presented and in Section~\ref{sec:constructions} lower bounds
are presented. Conclusion and problems for future research are given in Section~\ref{sec:conclude}.

\section{Upper Bounds for Subspace Packings}
\vspace{-.25ex}
\label{sec:upper}

All the basic bounds (upper and lower) on $\cA_q (n,k,t;\lambda)$ for $\lambda =1$ can be
generalized for $\lambda >1$. The most basic bounds are the packing bound and the Johnson bounds~\cite{EtZh18}.
The combination of the packing bound and the Johnson bound for $(n-1)$-subspaces
implies:
\begin{proposition}
\label{prop_combination_packing_johnson_hyperplane}
If $n$, $k$, $t$, and $\lambda$ are positive integers such that $1 \leq t< k< n$
and $1 \leq \lambda \leq \sbinomq{n-t}{k-t}$, then $\cA_q(n,k,t;\lambda)\le$
$$
 \underset{0\le x\le \cA_q(n-1,k,s;\lambda)}{\max} \min\left\{
x+\left\lfloor\frac{\lambda\sbinomq{n-1}{t}-x\sbinomq{k}{t}}{\sbinomq{k-1}{s}}\right\rfloor,
\left\lfloor\frac{q^n-1}{q^{n-k}-1}\cdot x\right\rfloor
\right\}.
$$
\end{proposition}

\subsection{Bounds based on Inequalities}

The first new upper bound is based on using inequalities similar to~\cite{MR3543542} which used it
for an application on mixed-dimension subspace codes. We first give a technical auxiliary result.

\begin{lemma}
\label{lemma_standard_equations_special1}
Let $a_i$ be a non-negative number for each integer $i \geq 0$. If there exist numbers $\mu_0,\mu_1,\mu_2$ and a
positive integer $m$ such that $\sum_{i\ge 0} a_i=\mu_0$,
$\sum_{i\ge 0} ia_i=\mu_1 c$, $\sum_{i\ge 0} i(i-1)a_i \leq \mu_2c$, and $2m\mu_1>\mu_2$ then $c\le \frac{m(m+1)\mu_0}{2m\mu_1-\mu_2}$.
\end{lemma}
\begin{proof}
Let $m$ be an arbitrary integer, then
$
m(m+1)\sum_{i\ge 0} a_i -2m \sum_{i\ge 0} ia_i$ $+ \sum_{i\ge 0} i(i-1)a_i  \leq m(m+1)\mu_0 -2m \mu_1 c + \mu_2c
$,
which implies that
$
\sum_{i\ge 0} (i-m)(i-m-1)a_i\le m(m+1)\mu_0-2m\mu_1c+\mu_2c
$.
Since $\sum_{i\ge 0} (i-m)(i-m-1)a_i \geq 0$, the last inequality is reduced to
$
0 \leq m(m+1)\mu_0-2m\mu_1c+\mu_2c
$,
which implies that
$
c \leq \frac{m(m+1)\mu_0}{2m\mu_1-\mu_2}
$.
\end{proof}
Minimizing the upper bound for $c$ in Lemma~\ref{lemma_standard_equations_special1}
as a function of $m$ induces $m=\frac{\mu_2\pm\sqrt{\mu_2^2+\mu_2}}{2\mu_1}$. Assuming $\mu_1>0$, $\mu_2\ge 0$,
the optimal choice would be $m=\frac{\mu_2+\sqrt{\mu_2^2+\mu_2}}{2\mu_1}$ since we have to satisfy $2m\mu_1>\mu_2$. Moreover, $m$ has
to be an integer, so that $m=\left\lceil \frac{\mu_2+\sqrt{\mu_2^2+\mu_2}}{2\mu_1} \right\rceil$ is a good choice.

\begin{proposition}
\label{prop_quadratic_bound_1}
If $2(q+1)m > \sbinomq{n-2}{1}$ for a positive integer $m$ and $n \geq 3$, then
$\cA_q(n,n-2,n-3;2)\le \left\lfloor \sbinomq{n}{1}\cdot\frac{m(m+1)}{2(q+1)m-\sbinomq{n-2}{1}}\right\rfloor$.
\end{proposition}
\begin{proof}
Let $\cC$ be a code with $\cA_q(n,n-2,n-3;2)$ codewords and for each $i \geq 1$ let $a_i$ denote
the number of $(n-1)$-subspaces (hyperplanes) of $\F_q^n$
containing exactly $i$ codewords of $\cC$. Since there are $\sbinomq{n}{1}$ distinct
$(n-1)$-subspaces we clearly have
$
\sum_{i\ge 0} a_i = \sbinomq{n}{1}
$.
Each codeword $X$ is an $(n-2)$-subspace and hence it is contained in $\sbinomq{2}{1}$ hyperplanes.
On the other hand summing the number of codewords in all the $(n-1)$-subspaces (with repetitions)
is $\sum_{i\ge 1} ia_i$ and hence we have
$
\sum_{i\ge 0} ia_i = \sbinomq{2}{1} \cA_q(n,n-2,n-3;2)
$.
The number of ordered pairs of codewords from $\cC$ which are contained in a given hyperplane~$H$
which contains exactly $i$ codewords is $i(i-1)$. Hence, the number of of ordered pairs of codewords
which are contained in the same hyperplane with $i$ codewords is $i(i-1) a_i$. Therefore, the number of
such ordered pairs in all $(n-1)$-subspaces of $\F_q^n$ is $\sum_{i\ge 0} i(i-1)a_i$.
For a given codeword $X$ of dimension $n-2$, the number of other codewords which intersect $X$
in an $(n-3)$-subspace is at most $\sbinomq{n-2}{n-3}=\sbinomq{n-2}{1}$ since any $(n-3)$-subspace can be contained in at
most $\lambda =2$ codewords. Each two codewords which are contained in the same $(n-1)$-subspace intersect
in exactly $(n-3)$-subspace. Hence, the number of ordered pair in all the hyperplanes is at most
$\sbinomq{n-2}{1}A_q(n,n-2,n-3;2)$. Therefore, we have
$
\sum_{i\ge 0} i(i-1)a_i  \leq \sbinomq{n-2}{1} \cA_q(n,n-2,n-3;2)
$.
Thus, we can apply Lemma~\ref{lemma_standard_equations_special1} with $m$, $\mu_0=\sbinomq{n}{1}$,
$\mu_1=\sbinomq{2}{1}=q+1$, and $\mu_2=\sbinomq{n-2}{1}$ and obtain the claim of the proposition. 
(Note that $2m\mu_1>\mu_2$.)
\end{proof}

\subsection{Upper Bounds based on $q^r$-Divisible Codes}
\label{subsec_divisible}

The Johnson bounds~\cite{EtZh18} can
be improved by using $q^r$-divisible codes~\cite{upper_bounds_cdc}. A \emph{$q^r$-divisible code}
is a linear block code in the Hamming scheme where all weights are divisible
by $q^r$. This family of codes has been introduced by Ward~\cite{ward1999introduction}.

\begin{lemma}(\cite[Lemma 4]{upper_bounds_cdc})
\label{lemma_is_divisible}
Let $\cP$ be the multiset of 1-subspaces generated from a
non-empty multiset of subspaces of $\F_q^n$ all having dimension at least $k \geq 2$ and let $\cH$ be an $(n-1)$-subspace of $\F_q^n$.
Then,
$
|\cP| \equiv |\cP \cap \cH| \pmod {q^{k-1}}
$.
\end{lemma}

If we form a generator matrix from the column vectors associated with $\cP$, i.e. one representative from each 1-subspace,
then the generated code will be a linear $q^{k-1}$-divisible code.
Let $c$ be a codeword of the code and $\cH$ be the corresponding hyperplane. Then,
$\operatorname{wt}(c) =|\cP| - |\cP \cap \cH|$, which is divisible by $q^{k-1}$.

Associating the multiset $\cP$ with a weight function $\omega$ that counts the multiplicity of every point of $\F_q^n$.
If $\lambda$ is an upper bound for $\omega$, we define
the $\lambda$-complement $\overline{\cP}$ of $\cP$ via the weight function $\lambda-\omega(\cP)$.
As shown in~\cite[Lemma 2]{upper_bounds_cdc} we also have
$|\overline{\cP}| \equiv |(\overline{\cP}\cap \cH)| \pmod {q^{k-1}}$ for every hyperplane $\cH$, i.e., a $q^{k-1}$-divisible
code of length $|\overline{\cP} |$ must exist.

As an example consider the following application of the Johnson bound:
$$
A_2(9,4,2;1)\le \left\lfloor\sbinomtwo{9}{1} A_2(8,3,1;1)/\sbinomtwo{4}{1} \right\rfloor=\left\lfloor\frac{17374}{15}\right\rfloor
=\left\lfloor1158+\frac{4}{15}\right\rfloor~.
$$
If $1158$ would be attained, then there would be a $2^3$-divisible code of length $4$. For cardinality
$1157$ there would be a $2^3$-divisible code of length $4+15=19$. Since no such codes exist, we have $A_2(9,4,2;1)\le 1156$.
Fortunately, the possible lengths of $q^r$-divisible codes over $\F_q$ have been
completely characterized in \cite{upper_bounds_cdc}. Each $t$-subspace
is $q^{t-1}$-divisible such that each $q^j$-fold copy of an $(t-j)$-subspace is $q^{t-1}$-divisible for all $0\le j<t$. Via concatenation
we see that there exists a $q^r$-divisible code of length $n=\sum_{i=0}^r a_i\cdot q^i\cdot\sbinomq{r+1-i}{1}$ for all $a_i\in\mathbb{N}_{\ge 0}$
for $0\le i\le r$. \cite[Theorem 4]{upper_bounds_cdc} states that a $q^r$-divisible code of length $n$ exists if and only if $n$ admits such
a representation as a non-negative integer linear combination of $q^i\cdot\sbinomq{r+1-i}{1}$ for $0\le i\le r$. Moreover, if
$n=\sum_{i=0}^r a_i\cdot q^i\cdot\sbinomq{r+1-i}{1}$ with $0\le a_i\le q-1$ for $0\le i\le r-1$ and $a_r<0$, then no $q^r$-divisible code
of length $n$ exists. In our example of $2^3$-divisible codes the possible summands are $15$, $14$, $12$, and $8$. The representations
$4=0\cdot 15+0\cdot 14+1\cdot 12-1\cdot 8$ and $19=1\cdot 15+0\cdot 14+1\cdot 12-1\cdot 8$ implies that no $2^3$-divisible
codes of lengths $4$ or $19$ exists.
We can reduce until
the remainder is a possible length of a $q^{k-1}$-divisible code. For this purpose we define

\begin{definition}
\label{def:newdef}
Let $\left\{a / \sbinomq{k}{1}\right\}_k$ denote the maximum $b\in\mathbb{N}$ for which
$a-b\cdot \sbinomq{k}{1}$ is a non-negative integer that is attained as length of some $q^{k-1}$-divisble code. 
\end{definition}

An efficient algorithm for the computation of $\left\{a / \sbinomq{k}{1} \right\}_k$ was given in~\cite{upper_bounds_cdc}.
The Johnson bound is improved as follows.
\begin{proposition}
\label{johnson_bound_point_inproved}
If $n$, $k$, $t$, and $\lambda$ are positive integers such that $1\le t\le k\le n$
and $1 \leq \lambda \leq \sbinomq{n-t}{k-t}$, then
$$A_q(n,k,t;\lambda)\le \left\{\sbinomq{n}{1} \cdot A_q(n-1,k-1,t-1;\lambda)/\sbinomq{k}{1} \right\}_k~.$$
\end{proposition}
\begin{proof}
Let $\cP$ be the $q^{k-1}$-divisible multiset of points of the
codewords, see Lemma \ref{lemma_is_divisible}. In $\cP$ every point has multiplicity at
most $A_q(n-1, k-1, s-1; \lambda)$ so that the $A_q(n-1, k-1, s-1; \lambda)$-complement is also
$q^{k-1}$-divisible. Thus, the claim follows from Definition~\ref{def:newdef}.
\end{proof}

Proposition~\ref{johnson_bound_point_inproved} gives
$
A_2(6,4,3;2) \leq  \left\{63\cdot A_2(5,3,2;2)/ 15\right\}_4=
\left\{63\cdot 32/ 15\right\}_4$ $=132
$,
while the Johnson bound only gives $A_2(6,4,3;2)\le 134$.
This specific bound is further improved since its parameters are small. For larger parameters
further such improvements are unknown.

\section{Constructions for Subspace Packings}
\vspace{-.25ex}
\label{sec:constructions}

The echelon-Ferrers construction (see~\cite{EtSt16} and references therein) and its generalizations
are probably the most successful constructions when we are given a set
of parameters $n$, $k$, $t$, and $\lambda$, such that $n/2 \geq k > t$.
These constructions are using rank-metric codes and in particular maximum rank distance (MRD in short)
codes~\cite{EtSi13,EtSt16} (we denote the rank-distance by $d_R$). But, there are some
other constructions that for some parameters are better than the echelon-Ferrers Construction.
The generalization of the \emph{linkage construction}~\cite{MR3543532,heinlein2017asymptotic} is one such example which is
not a straightforward generalization.
For small parameters the linkage construction is as good as the echelon-Ferrers Construction (see~\cite{heinlein2016tables}).

\subsection{A variant of the linkage construction}
\label{subsec_linkage}

An $\alpha-(n,k,\delta)_q^c$ covering Grassmanian
code $\cC$ consists of a set of $k$-subspaces of $\F_q^n$ such that every set of $\alpha$ codewords span
a subspace of dimension at least $\delta+k$. The maximum size of a related code is denoted by $\cB_q(n,k,\delta;\alpha)$.
It was proved in~\cite{EtZh18} that
$
\cA_q(n,k,t;\lambda)=\cB_q(n,n-k,k-t+1;\lambda+1)~,
$
and
$
\cB_q(n,k,\delta;\alpha)=\cA_q(n,n-k,n-k-\delta+1;\alpha-1)
$.

Finally, we will use a simple connection between the subspace distance of two $k$-subspaces $U$ and $V$ of $\F_q^n$,
and a related rank for the row space of these two subspaces
$
d_S (U,V)=2\dim(U+W)-\dim(U)-\dim(V)=2\left(\operatorname{rk}\begin{pmatrix}\tau(U)\\\tau(V)\end{pmatrix}-k\right)
$.
Here $\tau(U)$ and $\tau(V)$ are $k \times n$ matrices over $\F_q$ whose row spaces are $U$ and $V$.
Similarly, if $U$ and $V$ arise from lifting two matrices $M_1$ and $M_2$, then
$d_S(U,V)\ge 2\operatorname{rk}(M_1-M_2)=2 d_R (M_1,M_2)$.

\begin{theorem}
\label{thm:linkage}
	Let $1 \leq \delta \leq k$, $k+\delta\leq n$ and $2\leq \alpha \leq q^k+1$ be integers.
	\begin{enumerate}
		\item\label{1} If $n<k+2\delta,$ then
		$\cB_q(n,k,\delta ;\alpha)\geq (\alpha -1)q^{\max\{k,n-k\}(\min\{k,n-k\}-\delta+1)}$.
		
		\item\label{2} If $n\geq k+2\delta,$ then for each $t$ such that $\delta \leq t \leq n-k-\delta$, we have
		\begin{enumerate}
			\item\label{2a} If $t<k$, then			
			$\cB_q(n,k,\delta ;\alpha)\geq (\alpha -1)q^{k(t-\delta +1)} \cB_q(n-t,k,\delta ;\alpha)$.

			\item\label{2b} If $t\geq k$, then 			
$\cB_q(n,k,\delta ;\alpha)\geq (\alpha -1)q^{t(k-\delta +1)} \cB_q(n-t,k,\delta ;\alpha)+\cB_q(t+k-\delta,k,\delta;\alpha)$.
		\end{enumerate}
	\end{enumerate}
\end{theorem}
\begin{remark}
Note that the length of vectors is expected to be greater than or equal to $ k+\delta$. However, in Case \ref{2b} of Theorem \ref{thm:linkage}, there is a possibility that $t+k-\delta<k+\delta$ for $\cB_q(t+k-\delta,k,\delta;\alpha)$. In such situations, we consider the following convention:
$
\cB_q(t+k-\delta,k,\delta;\alpha)=\min\left\{\alpha-1, {t+k-\delta \brack k}_q  \right\}
$.
\end{remark}
The proof of Theorem~\ref{thm:linkage} will be in a few steps.
\noindent
{\bf Case \ref{1}: $k+\delta \leq n<k+2\delta$}
\begin{construction}
\label{const1}
Let $I_k$ denote the $k\times k$ identity matrix over $\mathbb{F}_q$ and let $C_1\subseteq \mathbb{F}_q^{k\times (n-k)}$
be a linear MRD code with minimum rank distance $\delta$. Let $C_1, C_2,\dots ,C_{\alpha -1}$
be $\alpha -1$ pairwise disjoint MRD codes with minimum rank
distance $\delta$ obtained by translating $C_1$ in a way that (see~\cite{EtSi13})
$
d_R(C_1\cup\dots\cup C_{\alpha -1})=\delta -1
$.
Let $C \triangleq C_1\cup\dots\cup C_{\alpha -1}$. Lifting the matrices in $C$,
$$(\alpha -1) q^{\max\{k,n-k\}(\min\{k,n-k\}-\delta +1)}$$ different matrices
of size $k\times n$, in reduced row echelon form (RREF in short), are constructed.
Let $\mathrm{RREF}(\C)$ denote the set of these matrices, and let $\C$
be the set of rowspaces of matrices in $\mathrm{RREF}(\C)$.
\end{construction}
\begin{claim}
Let $\C$ be the set of $k$-subspaces obtained in Construction~\ref{const1}. Then we have
$
\dim (U_1+\dots +U_{\alpha})\geq k+\delta
$,
for each $\alpha$ distinct codewords $U_1,\dots ,U_\alpha \in \C$.
\end{claim}
\begin{proof}
Given $\alpha$ distinct codewords $U_1,\dots,U_{\alpha}\in \C$,
let $u_1,\dots,u_{\alpha}\in \mathrm{RREF}(\C)$ be the corresponding $k\times n$ matrices
in RREF. Let $A_1,\dots ,A_\alpha$ be the $\alpha$ distinct codewords of $C$ satisfying
$
U_i
=\mathrm{rowspace}(I_k|A_i)
$
for each $1\leq i\leq \alpha$. For these $\alpha$ codewords of $\C$ we have that $\dim (U_1+\dots +U_{\alpha})$ is equal to
the rank of the $(\alpha k) \times n$ related matrix, i.e.
\begin{equation}\label{eq:rank-1}
\mathrm{rank}
\begin{tabular}{l}
\begin{tabular}{|l|l|}	
\hline $I_k$ & $A_1$\\
\hline $I_k$ & $A_2$\\
\hline $\vdots$ & \vdots \\
\hline $I_k$ & $A_\alpha$  \\ \hline
\end{tabular}
\end{tabular}. 
\end{equation}
Note that $A_1,\dots ,A_\alpha\in C=C_1\cup\dots\cup C_{\alpha-1}$, i.e. at least two of $A_i$'s
must be from the same rank-metric code $C_j$ for some $1\leq j\leq \alpha-1$. W.l.o.g., assume $A_1$ and $A_2$ are
from the same code $C_j$ for some $1\leq j\leq \alpha-1$. Clearly (\ref{eq:rank-1}) is equal to
	\[
	\mathrm{rank}
	\begin{tabular}{l}
	\begin{tabular}{|l|l|}	
	\hline $I_k$ & $A_1$\\
	\hline $0$ & $A_2-A_1$\\
	\hline $\vdots$ & \vdots \\
	\hline $0$ & $A_\alpha-A_1$  \\ \hline
	\end{tabular}
	\end{tabular}\geq \mathrm{rank}
	\begin{tabular}{l}
	\begin{tabular}{|l|l|}	
	\hline $I_k$ & $A_1$\\
	\hline $0$ & $A_2-A_1$\\ \hline
	\end{tabular}
	\end{tabular}\geq k+\delta.
	\]
\end{proof}

\noindent
{\bf Case \ref{2a}: $k+2\delta \leq n$, $t\leq n-k-\delta$, and $\delta \leq t < k$}
\vspace{-0.3cm}
\begin{construction}
\label{const2}
Let $\C_{n-t}$ be a set of $k$-subspaces of $\F_q^{n-t}$ such that any
$\alpha$ distinct $k$-subspaces $V_1,\dots ,V_{\alpha}\in\C_{n-t}$ satisfy $\dim (V_1+\dots +V_{\alpha})\geq k+\delta$,
and $|\C_{n-t}|=B_q(n-t,k,\delta ;\alpha)$ (note that $n-t\geq k+\delta$).
\begin{enumerate}
\item For each $V \in \C_{n-t}$, let $v \in \F_q^{k\times (n-t)}$ be the unique matrix in RREF such
that $V$ is the rowspace of $v$. The set $\mathrm{RREF}(\C_{n-t})$ contains all the subspaces of $\C_{n-t}$ in this form.
		
\item Let $C_1\subseteq \F_q^{k\times t}$ be a linear MRD code with minimum rank distance $\delta$.
Let $C_1,C_2,\dots ,C_{\alpha -1}$ be $\alpha -1$ pairwise disjoint MRD codes with minimum rank distance $\delta$
obtained by translating $C_1$ in a way that (see~\cite{EtSi13})
$$
d_R(C_1\cup\dots\cup C_{\alpha -1})=\delta -1.
$$
Let $C \triangleq C_1\cup\dots\cup C_{\alpha -1}$. By concatenating each matrix in $C$ to the end of each
$u\in\mathrm{RREF}(\C_{n-t})$, $(\alpha -1) q^{k(t-\delta+1)}|\C_{n-t}|$
different matrices, of size $k\times n$, in RREF are constructed. Let $\mathrm{RREF}(\C)$ denote the set of these matrices,
whose rowspaces form the code $\C$.
\end{enumerate}
\end{construction}
\begin{claim}
\label{claim2}
If $\C$ is the set of $k$-subspaces in Construction \ref{const2}, then
$
\dim (U_1+\dots +U_{\alpha})\geq k+\delta
$,
for each $\alpha$ distinct codewords $U_1,\dots ,U_{\alpha}$ of $\C$.
\end{claim}
\begin{proof}
Given $\alpha$ distinct codewords $U_1,\dots,U_{\alpha}$ of $\C$, let $u_1,\dots,u_{\alpha}\in \mathrm{RREF}(\C)$
be the corresponding $k\times n$ matrices in RREF. Let
$v_1,\dots ,v_{\alpha}\in \mathrm{RREF}(\C_{n-t})$ and $A_1,\dots ,A_\alpha$ be $\alpha$ codewords of $C$ satisfying
$$
U_i=\mathrm{rowspace}(u_i)=\mathrm{rowspace}([v_i|A_i])
$$
for each $1\leq i\leq \alpha$. Clearly, $\dim (U_1+\dots +U_{\alpha})$ is equal to
\begin{equation}
\label{eq:rank-2}
\mathrm{rank}
\begin{tabular}{l}
\begin{tabular}{|l|l|}	
\hline $v_1$ & $A_1$\\
\hline $v_2$ & $A_2$\\
\hline $\vdots$ & \vdots \\
\hline $v_\alpha$ & $A_\alpha$  \\ \hline
\end{tabular}
\end{tabular}~. 
\end{equation}
We distinguish between three cases.
\begin{itemize}
\item \textbf{Case A.} If $v_1=v_2=\dots =v_\alpha$, then $A_1,\dots ,A_\alpha$ are different matrices. Note that
$A_1,\dots ,A_\alpha\in C=C_1\cup\dots\cup C_{\alpha-1}$, which implies that at least two of the $A_i$'s must be from the
same rank-metric code $C_j$ for some $1\leq j\leq \alpha-1$. W.l.o.g., assume $A_1$ and $A_2$ are from
the code $C_j$ for some $1\leq j\leq \alpha-1$. Then clearly (\ref{eq:rank-2}) is equal to
$$
\mathrm{rank}
\begin{tabular}{l}
\begin{tabular}{|l|l|}	
\hline $v_1$ & $A_1$\\
\hline $0$ & $A_2-A_1$\\
\hline $\vdots$ & \vdots \\
\hline $0$ & $A_\alpha-A_1$  \\ \hline
\end{tabular}
\end{tabular}\geq \mathrm{rank}
\begin{tabular}{l}
\begin{tabular}{|l|l|}	
\hline $v_1$ & $A_1$\\
\hline $0$ & $A_2-A_1$\\ \hline
\end{tabular}
\end{tabular}\geq k+\delta.
$$
			
\item \textbf{Case B.} Assume $v_i\neq v_j$ for all $1\leq i<j\leq \alpha$. In this case,
			\[
			\begin{array}{lll}
			\mathrm{rank}
			\begin{tabular}{l}
			\begin{tabular}{|l|l|}	
			\hline $v_1$ & $A_1$\\
			\hline $v_2$ & $A_2$\\
			\hline $\vdots$ & \vdots \\
			\hline $v_\alpha$ & $A_\alpha$  \\ \hline
			\end{tabular}
			\end{tabular} 
			& \geq &
			\mathrm{rank}
			\begin{tabular}{l}
			\begin{tabular}{|l|}	
			\hline $v_1$ \\
			\hline $v_2$ \\
			\hline $\vdots$  \\
			\hline $v_\alpha$ \\ \hline
			\end{tabular}
			\end{tabular}\\ 
			\ & \ & \ \\
			\ & = & \dim (\mathrm{rowspace}(v_1)+\dots +\mathrm{rowspace}(v_\alpha)) 
            \geq  k+\delta
			\end{array}		
			\]
			by the definition of $\C_{n-t}$.
			
\item \textbf{Case C.} The only remaining case is that some of the $v_i$'s are different and some are equal.
W.l.o.g. assume that $v_1\neq v_2 =v_3$ which implies $A_2\neq A_3$. Hence,
equation (\ref{eq:rank-2}) equals to		
$$
\begin{array}{lll}
\mathrm{rank}
\begin{tabular}{l}
\begin{tabular}{|l|l|}	
			\hline $v_1$ & $A_1$\\
			\hline $v_2$ & $A_2$\\
			\hline $0$ & $A_3-A_2$\\
			\hline $\vdots$ & \vdots \\
			\hline $v_\alpha$ & $A_\alpha$  \\ \hline
\end{tabular}
\end{tabular}
& \geq &  \mathrm{rank}
\begin{tabular}{l}
\begin{tabular}{|l|l|}	
\hline $v_1$ & $A_1$\\
\hline $v_2$ & $A_2$\\
\hline $0$ & $A_3-A_2$\\
\hline
\end{tabular}
\end{tabular}
\\ \ & \geq & \mathrm{rank}
\begin{tabular}{l}
\begin{tabular}{|l|}	
\hline $v_1$ \\
\hline $v_2$ \\
\hline
\end{tabular}
\end{tabular} +\mathrm{rank}(A_3-A_2)
\\ \ & \ & \\ \ & \geq & (k+1)+(\delta -1) = k+\delta.
\end{array}
$$
\end{itemize}
\end{proof}
\vspace{-0.3cm}
\noindent
{\bf Case \ref{2b}: $k+2\delta \leq n$ and $k \leq t\leq n-k-\delta$}
\vspace{-0.4cm}
\begin{construction}
\label{const3}
Let $\C_{n-t}$ be a set of $k$-subspaces of $\F_q^{n-t}$ such that any
$\alpha$ distinct $k$-subspaces $U_1,\dots ,U_{\alpha}\in\C_{n-t}$ satisfy $\dim (U_1+\dots +U_{\alpha})\geq k+\delta$, and $|\C_{n-t}|=B_q(n-t,k,\delta ;\alpha)$ (note that $n-t\geq k+\delta$).
\begin{enumerate}
\item For each $U \in \C_{n-t}$, let $u \in \F_q^{k\times (n-t)}$ be the unique matrix in RREF such
that $U$ is the rowspace of $u$. The set $\mathrm{RREF}(\C_{n-t})$ contains all the subspaces of $\C_{n-t}$ in this form.

\item Let $C_1\subseteq \F_q^{k\times t}$ be a linear MRD code with minimum rank distance $\delta$.
Let $C_1,C_2,\dots ,C_{\alpha -1}$ be the $\alpha -1$ pairwise disjoint MRD codes of minimum rank distance $\delta$
obtained by translating $C_1$ in a way that (see~\cite{EtSi13})
$$
d_R(C_1\cup\dots\cup C_{\alpha -1})=\delta -1.
$$
Let $C \triangleq C_1\cup\dots\cup C_{\alpha -1}$. By concatenating each matrix in $C$ to the end of each matrix
$u\in\mathrm{RREF}(\C_{n-t})$, $(\alpha -1) q^{t(k-\delta+1)}|\cC_{n-t}|$
different matrices, of size $k\times n$, in RREF are constructed. Let $\mathrm{RREF}(\C)$ denote the set of these matrices,
whose rowspaces form the code $\C$.
		
\item Consider a code $\C_{\mathrm{app}}\subseteq \cG_q(n,k)$ such that
\begin{itemize}
\item the first $n-(t+k-\delta)$ entries of each codeword in $\C_{\mathrm{app}}$ are \emph{zeroes},
\item Each $\alpha$ distinct codewords $U_1,\dots,U_\alpha$ of $\C_{\mathrm{app}}$, satisfy $\dim (U_1+\dots+U_\alpha)\geq k+\delta$.
\item $\C_{\mathrm{app}}$ is of maximum size, i.e. $|\C_{\mathrm{app}}|=\cB_q(t+k-\delta,k,\delta;\alpha)$.
\end{itemize}		
\end{enumerate}
	
Form a new code $\C'$ as the union of $\C$ in Step 2 and $\C_{\mathrm{app}}$ in Step 3.
\end{construction}
\vspace{-0.5cm}
\begin{claim}
If $\C'$ is the set of $k$-subspaces in Construction~\ref{const3} and $U_1,\ldots,U_\alpha$ are $\alpha$
distinct codewords of $\C'$, then
$
\dim (U_1+\dots +U_{\alpha})\geq k+\delta
$.
\end{claim}
\vspace{-0.3cm}
\begin{proof}
The first two steps of Construction~\ref{const3} are the same as the ones in Construction~\ref{const2}. Therefore, the Claim
follows from the proof of the claim after Construction~\ref{const3} and the definition of $\C_{\mathrm{app}}$ in Construction~\ref{const3}.
\end{proof}
\vspace{-0.3cm}
\begin{corollary}
\label{cor:linkage}
	Let $1\leq s \leq k\leq n$ and $1\leq \lambda \leq q^k$ be integers.
	\begin{enumerate}
		\item\label{1} If $k>2t-2,$ then
		\[\cA_q(n,k,t ;\lambda)\geq \lambda q^{\max\{k,n-k\}(\min\{k,n-k\}-k+t)}. \]
		
		\item\label{2} If $k\leq 2t-2,$ then choosing an arbitrary $s$ satisfying $k-t+1 \leq s \leq t-1,$ we have that	
		\begin{enumerate}
			\item\label{2a} If $s<n-k$, then			
			\[\cA_q(n,k,t ;\lambda)\geq \lambda q^{(n-k)(s-k +t)} \cA_q(n-s,k-s,t-s ;\lambda).\]
			
			\item\label{2b} If $s\geq n-k$, then 			
			$$\cA_q(n,k,t ;\lambda)\geq \lambda q^{t(n-2k+t)} \cA_q(n-s,k-s,t-s;\lambda)$$
            $$+\cA_q(s+n-2k+t-1,s-k+t-1,s-2k-2t-1;\lambda).$$
		\end{enumerate}
	\end{enumerate}
\vspace{-0.4cm}
\end{corollary}

\vspace{-0.2cm}

\subsection{Integer Linear Programming Lower Bounds}
\vspace{-0.1cm}
\label{sec:IntProg}
The problem of the determination of $\cA_q(n,k,t;\lambda)$ can be formulated as an integer linear programming problem. For $\lambda=1$
reader is referred to~\cite{kohnert2008construction}. For each $k$-subspace $U$ of
$\F_q^n$ a binary variable $x_U$ is defined. The value of this variables is \emph{one} if $U$ is contained
in the subspace packing and \emph{zero} if $U$
is not contained in the subspace packing. The set of equations
contains a huge number of variables and constraints:
\begin{eqnarray*}
\max \sum_{U \in \cG_q(n,k)} x_U\\
\text{subject~to}\nonumber\\
\text{for~each} ~ V \in \cG_q(n,t) ~~~ \sum_{V \subset U \in \cG_q(n,k)} x_U &\le& \lambda\nonumber\\
1\le i<t ~ \text{and} ~ W \in \cG_q(n,i) ~~~ \sum_{W\le U\le \F_q^n\,:\, \dim(U)=k} x_U &\le& A_q(n-i,k-i,t-i;\lambda), \nonumber\\
\text{where} ~ x_U &\in& \{0,1\},~\text{for~each} ~ U \in \cG_q(n,k)~~~~~~~~~~~~~~~~~~~~~~~~~~~~~~~~~~\nonumber
\vspace{-0.1cm}
\end{eqnarray*}

The second set of constraints, i.e., those for $1\le i\le t-1$, are not necessary to guarantee that the maximum target value
equals $\cA_q(n,k,t;\lambda)$, but they may significantly speed up the computation.
However, this integer linear programming can be solved
for rather small parameters due to the exponential number of variables and constraints.
But, for small parameters some interesting bounds were obtained.

\vspace{-0.4cm}

\section{Discussion and Open Problems}
\label{sec:conclude}

\vspace{-0.2cm}

We have introduced new upper and lower bounds on $\cA_q(n,k,t;\lambda)$, the sizes of subspace packings.
In the extended version of this paper bounds for $t=1$, related to partial spread will
be given and also a few variants on the echelon-Ferrers construction. Some bounds
on specific parameters will be also given. At the end of this paper three
tables for specific lower and upper bounds are presented. Two interesting questions which are also related
to network coding are as follows:
\begin{itemize}
\item What are the asymptotic values of $\cA_q(n,k,t;\lambda)$?

\item What is the difference between sizes of the largest subspace packings,
with and without (any number of times) repeated codewords?
\end{itemize}

\vspace{-2.00ex}
%
%

\vspace{-1.00ex}
\bibliographystyle{plain}

\begin{thebibliography}{99}

\vspace{-0.2cm}

		
\bibitem{EtSi13}
    {\sc T. Etzion and N. Silberstein,}
    ``Codes and designs related to lifted MRD codes'',
    {\em IEEE Trans.\ Inform. Theory,} 59, 1004--1017, 2013.
\bibitem{EtSt16}
  {\sc T. Etzion and L. Storme,}
  ``Galois geometries and coding theory,''
  {\em Designs, Codes, and Cryptography,} 78, 311--350, 2016.






\bibitem{EtWa18}
    {\sc T. Etzion and A. Wachter-Zeh,}
    {\sl Vector network coding based on subspace codes outperforms scalar linear network coding,}
    {\em IEEE Trans. on Inform. Theory,} 64, 2460--2473, 2018.
\bibitem{EtZh18}
    {\sc T.~Etzion and H.~Zhang,}
    {\sl Grassmannian codes with new distance measures for network coding,}
    {\em arXiv preprint 1801.02329}, 2018.
\bibitem{MR3543532}
    {\sc H.~Gluesing-Luerssen and C.~Troha,}
    {\sl Construction of subspace codes through linkage,}
    {\em Advances in Mathematics of Communications}, 10, 525--540, 2016.
\bibitem{heinlein2016tables}
    {\sc D.~Heinlein, M.~Kiermaier, S.~Kurz, and A.~Wassermann,}
    {\sl Tables of subspace codes,}
    {\em arXiv preprint 1601.02864}, 2016.
\bibitem{heinlein2017asymptotic}
    {\sc D.~Heinlein and S.~Kurz,}
    {\sl Asymptotic bounds for the sizes of constant dimension codes and an improved lower bound,}
    in {\em Castle Meeting on Coding Theory and
      Applications}, 163--191, Springer, 2017.
\bibitem{MR3543542}
    {\sc T.~Honold, M.~Kiermaier, and S.~Kurz,}
    {\sl Constructions and bounds for mixed-dimension subspace codes,}
    {\em Adv. Math. Commun.}, 10, 649--682, 2016.
\bibitem{upper_bounds_cdc}
    {\sc M.~Kiermaier and S.~Kurz,}
    {\sl An improvement of the {J}ohnson bound for subspace codes,}
    {\em arXiv preprint 1707.00650}, 2017.
\bibitem{kohnert2008construction}
    {\sc A.~Kohnert and S.~Kurz,}
    {\sl Construction of large constant dimension codes with a prescribed minimum distance,}
    in {\em Mathematical methods in computer science}, 31--42, Springer, 2008.
\bibitem{MiMu92}
     {W. H. Mills and R. C. Mullin,}
     {\sl Coverings and packings}
     in {\em Contemporary Design Theory: A Collection of Surveys},
     editors {\sc J. H. Dinitz and D. R. Stinson,}
     {\sl John Wiley: New York, 1992}.
\bibitem{SWY06}
     {D. R. Stinson, R. Wei, and J. Yin,}
     {\sl Packings}
     in {\em The CRC Handbook of Combinatorial Designs},
     edited by {\sc C. J. Colbourn and J. H. Dinitz}
     {\sl John Wiley: New York, 2006}.
\bibitem{ward1999introduction}
    {\sc H.~N. Ward,}
    {\sl An introduction to divisible codes,}
    {\em Designs, Codes and Cryptography}, 17, 73--79, 1999.

\end{thebibliography}


\begin{samepage}
\begin{table}[htp]
  \begin{center}
    \begin{tabular}{|c|ccccc|}
    \hline
    k/t & 1 & 2 & 3 & 4 & 5\\
    \hline
    2 & $42$ &$651$ &         &        &       \\
    3 & $18$ & $180$ & $1395$ &        &       \\
    4 & $6$  &  $21$ & $121-126$  & $651$ &       \\
    5 & $2$  &  $2$ & $2$   & $32$  & $63$ \\
    \hline
    \end{tabular}
\medskip
    \caption{Bounds for $A_2(6,k,t;2)$}
  \end{center}
\end{table}
\begin{table}[htp]
\vspace{-1.2cm}
  \begin{center}
    \begin{tabular}{|c|cccccc|}
    \hline
    k/t & 1 & 2 & 3 & 4 & 5 & 6\\
    \hline
    2 & $84$ & $2667$ & & & & \\
    3 & $34$ & $741-762$ & $2667$ & & & \\
    4 & $16$ & $96-144$ & $906-1524$ & $11811$ & & \\
    5 & $2$ & $7$ & $43-85$ & $360-514$ & $2667$ & \\
    6 & $2$ & $2$ & $2$ & $2$ & $64$ & $127$ \\
    \hline
    \end{tabular}
\medskip
    \caption{Bounds for $A_2(7,k,t;2)$}
  \end{center}
\end{table}
\begin{table}[htp]
\vspace{-1.2cm}
  \begin{center}
    \begin{tabular}{|c|ccccccc|}
    \hline
    k/t & 1 & 2 & 3 & 4 & 5 & 6 & 7\\
    \hline
    2 & $170$ & $10795$ & & & & & \\
    3 & $72$ & $2663-3060$ & $97155$ & & & & \\
    4 & $34$ & $512-578$ & $6933-12954$ & $200787$ & & & \\
    5 & $10-11$ & $33-128$ & $318-1184$ & $4821-12532$ & $97155$ & & \\
    6 & $2$ & $2$ & $17-25$ & $71-341$ & $969-2078$ & $10795$ & \\
    7 & $2$ & $2$ & $2$ & $2$ & $2$ & $128$ & $255$ \\
    \hline
    \end{tabular}
\medskip
    \caption{Bounds for $A_2(8,k,t;2)$}
  \end{center}
\end{table}
\end{samepage}

\vspace{-0.6cm}

\end{document}